\documentclass[12pt,reqno]{amsart}

\usepackage{amsmath}
\usepackage{latexsym}
\usepackage{amsfonts}
\usepackage{amssymb}
\usepackage{color}
\usepackage{bbm,dsfont}
\usepackage{graphicx}
\usepackage{MnSymbol}
\usepackage{hyperref}
\usepackage{caption}

\usepackage{enumerate}

\newtheorem{theorem}{\bf Theorem}[section]

\newtheorem{proposition}{\bf Proposition}[section]

\newcommand{\R}{\mathbb{R}} 
\newcommand{\C}{\mathbb{C}} 
\newcommand{\real}{\mathbb R} 

\newcommand{\half}{\tfrac{1}{2}} 
\newcommand{\mo}[1]{\left| #1 \right|} 
\newcommand{\floor}[1]{\lfloor #1 \rfloor} 

\newcommand{\abs}{\mo} 

\newcommand{\xx}{\mathcal{X}}
\newcommand{\rr}{\mathcal{R}}

\newcommand{\hi}{\mathcal{H}} 
\newcommand{\hh}{\mathcal{H}} 
\newcommand{\lh}{\mathcal{L(H)}} 
\newcommand{\lhs}{\mathcal{L}_s(\hi)} 

\newcommand{\kb}[2]{|#1\rangle\langle#2|} 
\newcommand{\no}[1]{\left\|#1\right\|} 
\newcommand{\tr}[1]{{\rm tr}\left[#1\right]} 

\newcommand{\id}{\mathbbm{1}} 

\renewcommand{\rho}{\varrho}
\newcommand{\lam}{\lambda}

\newcommand{\rank}{{\rm rank}\,} 
\newcommand{\rankbis}{{\rm rank}} 

\newcommand{\va}{\mathbf{a}} 
\newcommand{\vb}{\mathbf{b}} 

\newcommand{\vr}{\mathbf{r}} 

\newcommand{\vsigma}{\boldsymbol{\sigma}} 

\newcommand{\Ac}{\mathcal{A}}

\newcommand{\prem}{\mathcal{P}}
\newcommand{\state}{\mathcal{S}} 

\newcommand{\tzero}{\mathcal{T}_0} 

\title[]{Probing quantum state space: \\
does one have to learn everything to learn something?}

\begin{document}\setlength{\arraycolsep}{2pt}

\begin{abstract}
Determining the state of a quantum system is a consuming procedure.
For this reason, whenever one is interested only in some particular property of a state, it would be desirable to design a measurement setup that reveals this property with as little effort as possible. 
Here we investigate whether, in order to successfully complete a given task of this kind, one needs an informationally complete measurement, or if something less demanding would suffice. 
The first alternative means that in order to complete the task, one needs a measurement which fully determines the state. We formulate the task as a membership problem related to a partitioning of the quantum state space and, in doing so, connect it to the geometry of the state space. For a general membership problem we prove various sufficient criteria that force informational completeness, and we explicitly treat several physically relevant examples. For the specific cases that do not require informational completeness, we also determine bounds on the minimal number of measurement outcomes needed to ensure success in the task. 
\end{abstract}

\author[Carmeli]{Claudio Carmeli}
\address{\textbf{Claudio Carmeli};DIME, Universit\`a di Genova, Via Magliotto 2, I-17100 Savona, Italy}
 \email{claudio.carmeli@gmail.com}

\author[Heinosaari]{Teiko Heinosaari}
\address{\textbf{Teiko Heinosaari}; Turku Centre for Quantum Physics, Department of Physics and Astronomy, University of Turku, FI-20014 Turku, Finland}
\email{teiko.heinosaari@utu.fi}

\author[Schultz]{Jussi Schultz}
\address{\textbf{Jussi Schultz}; Turku Centre for Quantum Physics, Department of Physics and Astronomy, University of Turku, FI-20014 Turku, Finland}
\email{jussischultz@gmail.com}

\author[Toigo]{Alessandro Toigo}
\address{\textbf{Alessandro Toigo}; Dipartimento di Matematica, Politecnico di Milano, Piazza Leonardo da Vinci 32, I-20133 Milano, Italy, and I.N.F.N., Sezione di Milano, Via Celoria 16, I-20133 Milano, Italy}
\email{alessandro.toigo@polimi.it}

\maketitle

\section{Introduction.}

When a quantum system undergoes a particular preparation procedure, it is left in a well-defined quantum state. In order to gain information about the state, one must perform measurements on the system. With a sufficiently large number of measurement runs it is possible to perform full state tomography, that is, to reconstruct the state and thus reveal all of the system's properties \cite{QSE04}. A measurement set-up which is sufficient for this task is called {\em informationally complete} \cite{Prugovecki77}. However, full state tomography is a consuming task, and in many cases one might well be satisfied with only partial information about the system. 
For instance, when an experimental set-up is built with the aim of preparing a specific quantum state, one may want to check that the device works as planned up to a certain error threshold. 
This task amounts to designing a measurement set-up that simply verifies if the prepared state is close to the target state with respect to a given distance measure.
Another aim could be to check the quality of the preparation procedure by estimating some bound on the von Neumann entropy of the quantum state, but not aiming to know anything more. 

Since these tasks are simple yes/no-questions, one might hope to achieve the answers with a measurement with two outcomes, or at least with something less than an informationally complete set-up. However, the formalism of quantum theory puts limitations on the possible measurements, so it is far from obvious that this can be done. Indeed, it was recently shown that the binary question of whether a bipartite state is entangled or separable requires an informationally complete measurement \cite{CaHeKaScTo16,Luetal16}.

In this work, we study the problem of deciding which one among finitely many mutually exclusive properties a quantum system possesses. We formulate this as a membership problem: each of the properties corresponds to a subset $\prem_j$ of the state space $\state$ of the system, and we want to know what kind of measurements allow us to infer which subset an unknown state belongs to. In particular, our main focus is on the following structural question on the duality of states and measurements:
\begin{quote}
\emph{What properties of an unknown quantum state can be inferred without resorting to full state tomography?}
\end{quote}

We consider this problem on the general level, as well as through some physically relevant examples. The general results of Section~\ref{sec:membership} highlight the role played by the geometry of the subsets of states representing the properties, and we find many conditions which force the successful set-ups to be informationally complete. In Section~\ref{sec:specific} we consider the problem for specific examples such as the quality control problems outlined above.
For the cases that do not require full state reconstruction, we also obtain bounds on the minimal number of outcomes that the  optimal measurements must have. 
These are discussed in Section~\ref{sec:minimal}. 
The conclusion are presented in Section~\ref{sec:conclusion}.

\section{Preliminaries.}\label{sec:prelim}

\subsection{Quantum state space.}

he mathematical description of a $d$-level quantum system is based on the Hilbert space $\hi=\C^d$. We denote by $\lh$ the space of linear operators on $\hi$, and by $\lhs$ the real subspace of selfadjoint (i.e.~Hermitian) operators. The states of the system are represented by positive unit trace operators, and we denote by $\state=\{ \varrho\in\lhs : \varrho\geq 0, \tr{\varrho}=1 \}$ the state space of the system. The state space is therefore a convex subset of $\lhs$, and its affine hull, i.e. the smallest affine space containing it,
is  the  space of selfadjoint trace one operators. 
The state space is equipped with the natural topology coming from the trace norm $\Vert A\Vert_1  = \tr{\sqrt{A^*A}}$, and it is compact in $\lhs$ with respect to this topology.  
In particular \(\state\) is a compact convex set with nonempty interior in its affine hull, and hence can be viewed as a convex body.
At some points, we are also going to use the Hilbert-Schmidt norm $\Vert A\Vert_2  = \sqrt{\tr{A^*A}}$.

We say that a state $\varrho\in\state$ belongs to the (geometric) {\em interior} of the state space, denoted by ${\rm Int}\, \state$, if for each $\sigma\in\state$ there exists a $\sigma'\in\state$ and $t\in(0,1)$ such that $\varrho = t\sigma + (1-t)\sigma'$.
The geometric interior of $\state$ coincides with the usual notion of relative interior of a convex set (see e.g.~\cite[Theorem 6.4]{Rockafellar}), which is the topological interior of $\state$ when viewed as a subset of its affine hull. 
It is easy to see that a state $\varrho$ is in the interior if and only if it has full rank, i.e., ${\rm rank}\, \varrho = d$ (see e.g.~\cite{GQS06}).

The complement of the interior with respect to the state space is called the {\em boundary} of the state space, and we denote it by $\partial \state$. A state is thus on the boundary if and only if it has $0$ as an eigenvalue. 
A special class of boundary states consists of extreme elements of $\state$, and these are called \emph{pure states}. 
Pure states correspond to one-dimensional projections.

\subsection{Quantum measurements.}

A measurement performed on a quantum system is represented by a positive operator valued measure (or POVM for short) \cite{MLQT12}.  A POVM with finitely many  outcomes is a map $E$ that assigns a positive operator $E_j$ to each outcome $j$ and satisfies the normalization condition $\sum_j E_j = \id$. Given a system in a state $\varrho\in\state$, the outcomes of a statistical experiment are then distributed according to the probabilities $p_j = \tr{\varrho E_j}$. 

Since we are interested on the amount of information about the state that can be extracted in a statistical experiment, we can switch to a slightly more flexible mathematical description and ignore the specific outcomes of a measurement. 
Indeed, using linearity we can instead consider the real linear span of the POVM, $\rr_E = \{ \sum_j r_j E_j  : r_j\in\R\}$. 
The numbers $\tr{\varrho A}$, $A\in\rr_E$, can be calculated from the probabilities $p_j$, and vice versa. The subspace $\rr_E\subseteq \lhs$ is closed under taking adjoints and contains the identity; hence it is an operator system (see e.g. \cite[Chapter 2]{CBMOA03}).
Any operator system can be generated by a POVM \cite{HeMaWo13}, but we can equally well consider any other collection $\Ac=\{A_1,A_2,\ldots, A_n\}$ of selfadjoint operators which, together with the identity $\id$, generate the same operator system. 
The numbers $\tr{\varrho A_j}$ then correspond to expectation values of measurements rather than just probabilities.

One can think of an operator system as an equivalence class of POVMs; two POVMs are equivalent if they span the same operator system. 
The questions we will investigate do not depend on the specific POVM but only on the generated operator system.
Hence, by a quantum measurement we will mean either a POVM or an operator system, this leading to no confusion. 

\subsection{Distinguishing states}\label{subsec:prelimC}

In a measurement of a POVM $E$, two states $\varrho_1,\varrho_2\in\state$ give the same outcome probability distribution if and only if $\tr{(\varrho_1-\varrho_2)E_j}=0$ for all outcomes $j$. By linearity, this means that the selfadjoint operator $\Delta = \varrho_1-\varrho_2$ is orthogonal to the operator system $\rr_E$ with respect to the (real) Hilbert-Schmidt inner product $\langle S\mid T\rangle_2 =\tr{ST}$. More generally, when we describe a measurement setting in terms of an operator system $\rr$, it is the orthogonal complement $\rr^\perp$ in $\lhs$ that describes the measurement's ability to distinguish between states. Since $\rr$ contains the identity operator $\id$, the orthogonal complement $\rr^\perp$ is a subspace of the real vector space of traceless selfadjoint operators, denoted by $\tzero$. Conversely, for any subspace $\xx\subseteq\tzero$, the orthogonal complement $\xx^\perp$ of $\xx$ in $\lhs$ is an operator system and thus generated by some POVM (see, for instance, \cite[Proposition 1]{HeMaWo13}). 

A measurement that can distinguish between any two different states is called {\em informationally complete} \cite{Prugovecki77,BuLa89}. 
In terms of operator systems, this means that $\rr=\lhs$ \cite{Busch91,SiSt92}, or equivalently, $\rr^\perp = \{0\}$. 
Indeed, if $\rr^\perp$ contains a nonzero operator $\Delta$, then we can write $\Delta = \Delta_+ - \Delta_-$, where $\pm\Delta_\pm$ are the positive and negative parts of $\Delta$ as given by the spectral decomposition. Since $\tr{\Delta} = 0$, we have $\tr{\Delta_+} = \tr{\Delta_-} \equiv \lambda >0$, hence $\Delta = \lambda (\varrho_+ - \varrho_-)$, for the two different states $\varrho_\pm = \Delta_\pm / \lambda$. But this means precisely that $\varrho_+$ and $\varrho_-$ are indistinguishable by the measurement in question.

\section{Membership problem}\label{sec:membership}

\subsection{General formulation}\label{subsec:general}

We consider a quantum system which possesses one of $N$ mutually exclusive properties. Mathematically, each property corresponds to a subset $\prem_j$ of the state space, and the task is to identify which subset a given state belongs to. We assume that the system possesses one of the properties, so that in general our  membership problem is defined by specifying a partitioning $\state= \cupdot_j \prem_j  $. 
Here the symbol $\cupdot$ underlines the fact that we are assuming the subsets \(\prem_j\) to be mutually disjoint, reflecting the fact that the properties are mutually exclusive. Our task is to infer the index $j$ for which $\varrho\in\prem_j$ by performing some measurement, and the main question is whether or not this can be done without requiring the measurement to be informationally complete.

An operator system $\rr$ solves the membership problem $\state= \cupdot_j \prem_j$ if and only if for all $\varrho_j\in\prem_j$ and $\varrho_k\in\prem_k$ with $j\neq k$ we have $\varrho_j-\varrho_k\notin \rr^\perp$, as otherwise those two states would be indistinguishable by $\rr$ and it would not be possible to infer their respective properties.
In particular, the membership problem can be solved {\em without} informational completeness if and only if we can find an $\rr$ such that $\rr^\perp \neq \{0 \}$, but which still succeeds in the task. Decreasing the size of $\rr^\perp$ (i.e., adding more measurement outcomes) does not affect the condition $\varrho_j-\varrho_k\notin \rr^\perp$. Hence, it is enough to study the cases where $\rr^\perp$ is one-dimensional. By turning this around, we conclude that informational completeness {\em is needed} if and only if any measurement with one-dimensional space $\rr^\perp$ is such that $\varrho_j-\varrho_k\in \rr^\perp$ for some $j\neq k$ and $\varrho_j\in\prem_j$ and $\varrho_k\in\prem_k$.

The membership problem can be understood intuitively as a question about the geometry of the corresponding subsets of states \cite{CaHeKaScTo16}. Indeed, by writing $\rr^\perp = \R\Delta$ for the one-dimensional space in the previous paragraph, the relevant question is whether or not some nonzero traceless operator $\Delta$ allows a decomposition of the form $\Delta=\lambda (\varrho_j-\varrho_k)$, with $\lam\in\R$ and the states coming from different subsets $\prem_j$ and $\prem_k$.
The following theorem now characterizes the membership problems which require informational completeness. 
For any two subsets $\prem_j,\prem_k \subset \state$, we denote $\R(\prem_j-\prem_k) = \{ \lam(\varrho_j-\varrho_k) : \varrho_j \in \prem_j,\ \varrho_k \in \prem_k,\ \lam\in\R \}$. Moreover, recall that $\tzero = \{ \Delta \in\lhs : \tr{\Delta} = 0 \}$. 

\begin{theorem}\label{thm:learn}
A membership problem $\state=\cupdot_j \prem_j$ cannot be solved without an informationally complete measurement if and only if
\begin{align*}
\bigcup_{j\neq k}\real(\prem_j - \prem_k) = \tzero \, .
\end{align*}
\end{theorem}

\begin{proof}
We have seen that the membership problem requires an informationally complete measurement if and only if, for any $\Delta\in\tzero$, there exist two states $\varrho_j\in\prem_j$ and $\varrho_k\in\prem_k$ with $j\neq k$, and a $\lambda\in\R$ such that $\Delta = \lambda (\varrho_j - \varrho_k)$. 
The indices $j$ and $k$ may depend on the operator $\Delta$, hence this is equivalent to $\tzero\subseteq\bigcup_{j\neq k}\real (\prem_j - \prem_k)  $.
The other inclusion is trivial, thus the two sets are equal.
\end{proof}

In the present context, it is natural to call any nonzero operator $\Delta\in\tzero$ a {\em perturbation operator}. This terminology is further justified by Fig.~\ref{Fig:Data0}, which serves as a graphical illustration of our previous discussion. Moreover, Proposition \ref{prop:learn-2} below formalizes the idea for a membership problem involving two blocks, $\prem$ and $\prem^C$. As explained at the end of Section \ref{sec:prelim}.\ref{subsec:prelimC} by decomposing $\Delta$ into its positive and negative parts, \emph{any perturbation operator is, up to a nonzero scalar multiple, a difference of two distinct states}.

\begin{figure}[!h]
   \begin{minipage}{0.40\textwidth}
     \centering
     \includegraphics[width=.9\linewidth]{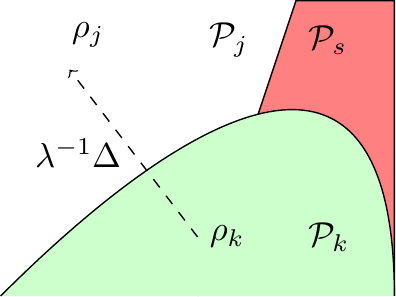}
         \caption{}\label{Fig:Data0}
   \end{minipage}\hfill
   \begin {minipage}{0.60\textwidth}
{F{\Small{IGURE}}\,1. In geometric terms, we can interpret a perturbation operator $\Delta$ as a direction in the state space, and we need to see if in that direction it is possible to cross the boundary between blocks $\prem_j$ and $\prem_k$. Indeed, solving $\Delta = \lam(\rho_j-\rho_k)$ for $\varrho_j = \varrho_k + \lambda^{-1}\Delta$, the question in Theorem \ref{thm:learn} then reduces to the existence of some $\varrho_k\in\prem_k$ such that, by adding a small perturbation $\lambda^{-1}\Delta$, we cross the boundary and obtain a state $\varrho_j\in\prem_j$. }
    \end{minipage}
\end{figure}

\begin{proposition}\label{prop:learn-2}
A membership problem $\state=\prem \cupdot \prem^C$ cannot be solved without an informationally complete measurement if and only if for any perturbation operator $\Delta$, there is a state $\varrho \in \prem$ and a $\lambda\in\R$ such that $\varrho + \lambda \Delta \in \prem^C$.
\end{proposition}

\begin{proof}
According to Theorem~\ref{thm:learn}, the necessity of informational completeness is equivalent to the condition $\R(\prem - \prem^C) =\tzero$, which means that any perturbation operator $\Delta$ can be written as $\Delta = \lam'(\varrho-\varrho')$ for some nonzero $\lam'\in\R$ and $\varrho\in\prem$, $\varrho'\in\prem^C$. 
Solving for $\varrho' =\varrho - \lambda'^{-1} \Delta$ and setting $\lam=-\lam'^{-1}$, the claim follows.
\end{proof}

\begin{figure}[!h]
\begin{center}
(a) \includegraphics[width=5.2cm]{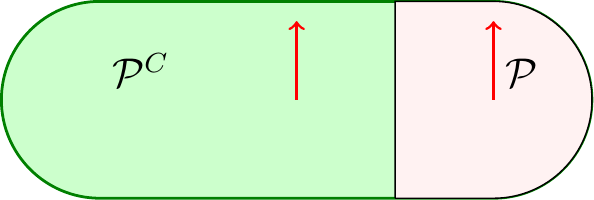}\hskip0.5cm (b) \includegraphics[width=5.2cm]{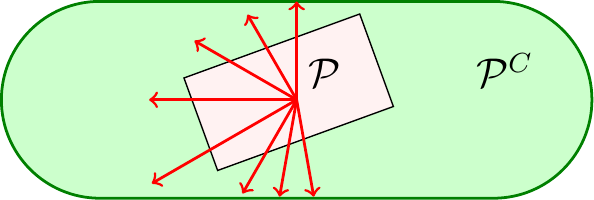}
\end{center}
\caption{\label{Fig:Data} 
Two partitions where (a) does not require informational completeness while (b) does.}
\end{figure}

The intuitive content of Proposition \ref{prop:learn-2} is that a membership problem $\state=\prem \cupdot \prem^C$ can be solved without an informationally complete measurement if and only if there is at least one direction such that no line parallel to it touches both $\prem$ and $\prem^C$. 
This is schematically depicted in Fig. \ref{Fig:Data}.

\subsection{Qubit membership problem with two blocks}

The simplest class of membership problems are those where the qubit state space is partitioned into two blocks, $\prem$ and $\prem^C$. 
 In the Bloch ball representation all qubit states are written as $\varrho_{\vr} = \half (\id + \vr \cdot \vsigma)$, $\vr \in \R^3$, $\no{\vr}\leq 1$, where $\vr \cdot \vsigma = r_x \sigma_x + r_y \sigma_y + r_z \sigma_z$.
Therefore, we can think of $\prem$ as a subset of the unit ball in $\R^3$. 
All perturbation operators are of the form $\Delta_{\va}=\va \cdot \vsigma$ for some nonzero $\va\in \R^3$.

Suppose that the membership problem $\state=\prem\cupdot \prem^C$ can be solved without an informationally complete measurement. By Proposition \ref{prop:learn-2}, this means that there exists a perturbation operator $\Delta_\va$ such that, if $\prem$ contains a state $\varrho_{\vr}$, then it must also contain all states $\varrho_{\vr} + \lam\Delta_\va = \varrho_{\vr + \lambda \va}$, and there are no other conditions. 
We have thus obtained the following geometric characterization of those qubit membership problems with two blocks that can be solved without an informationally complete measurement.

\begin{proposition}\label{prop:qubit}
A qubit membership problem $\state=\prem \cupdot \prem^C$ can be solved without an informationally complete measurement if and only if $\prem$ is an intersection of the Bloch ball with a family of parallel lines. 
\end{proposition}

\noindent Two illustrative examples of Proposition \ref{prop:qubit} are depicted in Fig. \ref{Fig:Qubit}.

\begin{figure}[!htb]
\begin{center}
(a) \includegraphics[width=4.5cm]{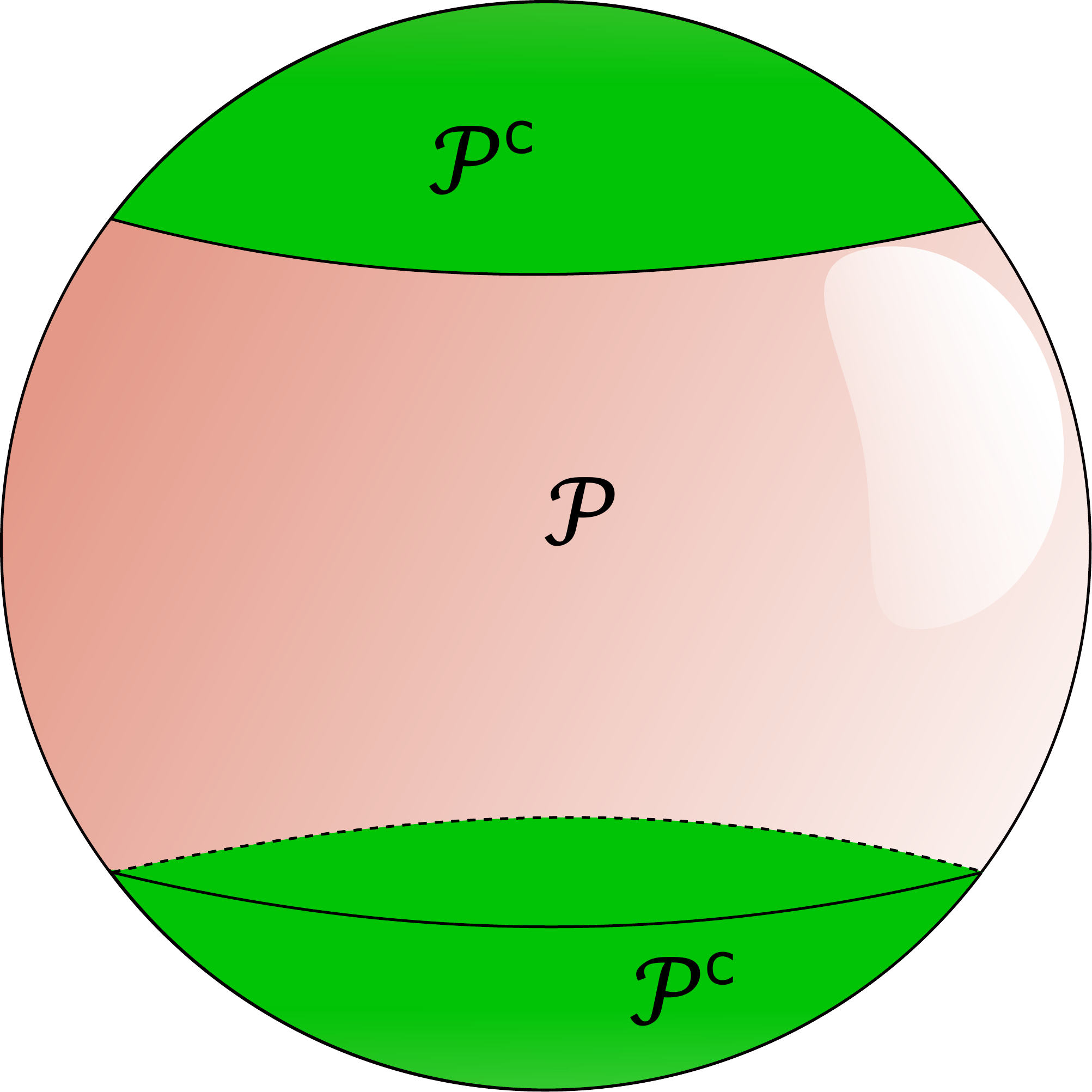}\hskip1cm (b) \includegraphics[width=4.5cm]{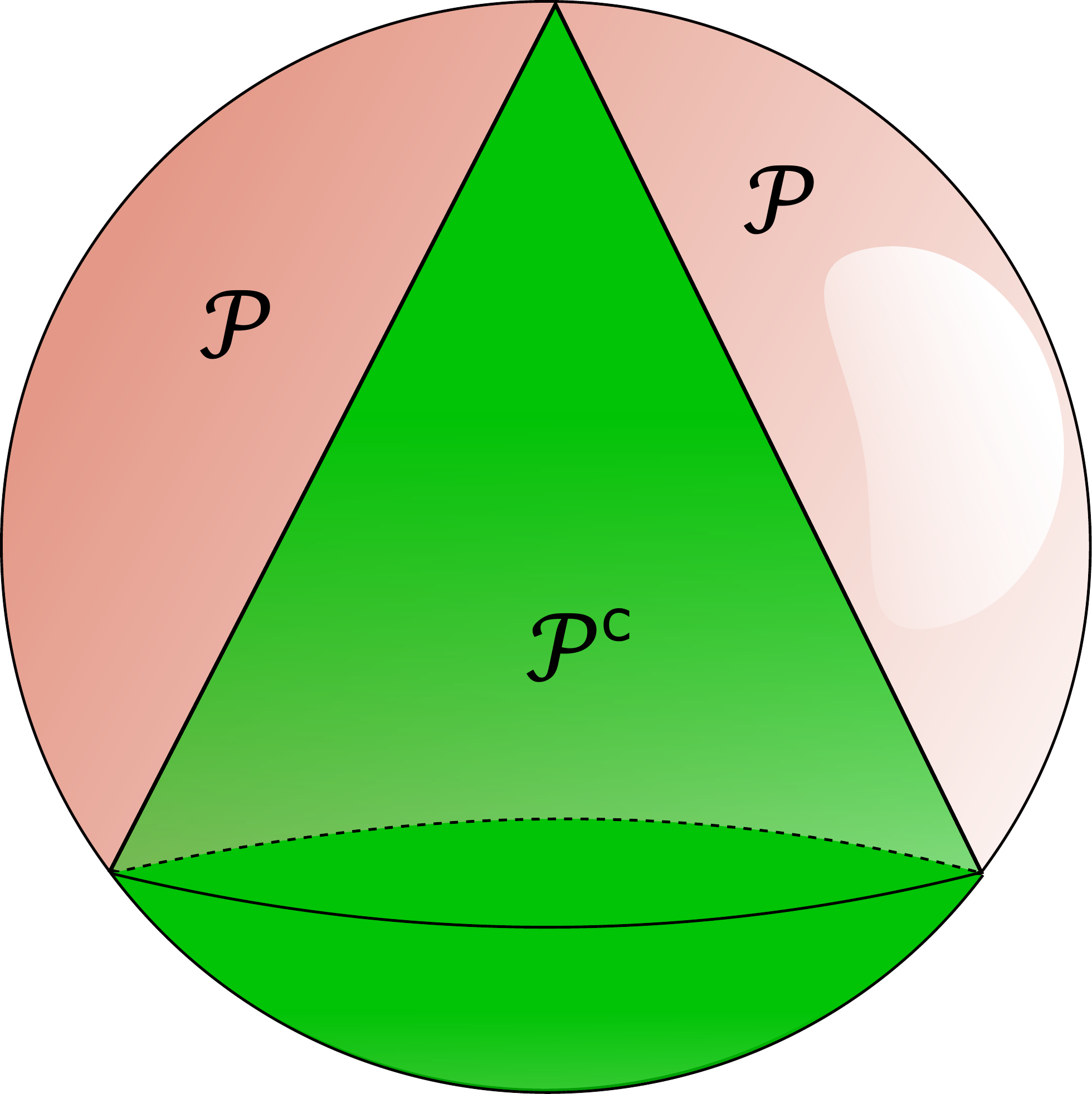}
\end{center}
\caption{\label{Fig:Qubit} Two Bloch ball partitions where the separations are formed with (a)  parallel lines and (b) non-parallel lines. The membership problem (a) can be solved without an informationally complete observable, while (b) cannot. }
\end{figure}

\subsection{Boundary criterion.}

A simple example of a membership problem that requires informational completeness is given by the partitioning of the state space into its interior and boundary points, $\state = {\rm Int}\, \state \cupdot \partial\state$. This is rather obvious from the geometric point of view: since the state space can be viewed as a convex body, any direction in it can be written as a difference of an interior point and a boundary point. 

This situation can be generalized a bit; it is enough that one of the blocks is contained in ${\rm Int}\, \state$. 
Recall that the interior ${\rm Int}\, \state$ consists exactly of states with full rank.

\begin{proposition}\label{prop:boundary}
Let $\state= \cupdot_j \prem_j$.
If one of the blocks $\prem_j$ contains only interior points of $\state$, then the membership problem cannot be solved without an informationally complete measurement.
\end{proposition}

\begin{proof}
Suppose $\prem_j\subseteq {\rm Int}\,\state$. 
Let $\Delta$ be a perturbation operator and let $\varrho_1\in\prem_j$.  
We denote by $\lam_{\rm min}$ the smallest eigenvalue of the operator $\sqrt{\varrho_1^{-1}}\Delta\sqrt{\varrho_1^{-1}}$ and define 
\begin{align*}
\varrho_2 =\sqrt{\varrho_1}(\id-\lam_{\rm min}^{-1}\sqrt{\varrho_1^{-1}}\Delta\sqrt{\varrho_1^{-1}})\sqrt{\varrho_1}.
\end{align*}
Since $\lam_{\rm min}<0$, the largest eigenvalue of the operator $\lam_{\rm min}^{-1}\sqrt{\varrho_1^{-1}}\Delta\sqrt{\varrho_1^{-1}}$ is $1$. 
It follows that \(\id-\lam_{\rm min}^{-1}\sqrt{\varrho_1^{-1}}\Delta\sqrt{\varrho_1^{-1}}\) is a positive operator with the smallest eigenvalue being \(0\). Since \(\tr{\varrho_2}=1\),  the operator $\varrho_2$ is a state; moreover, $\varrho_2$ has \(0\) as an eigenvalue, hence it is in the boundary $\partial\state$ of $\state$. In particular, since $\prem_j$ contains only interior points, we have $\varrho_2\in\prem_k$ for some $k\neq j$. 
Since $\lam_{\rm min}(\varrho_1 - \varrho_2) = \Delta$, the claim follows by Theorem \ref{thm:learn}.
\end{proof}

\subsection{Strictly convex block.}

In general, the convexity of a block $\prem$ has no implications on the corresponding membership problem $\state=\prem\cupdot\prem^C$. 
On the one hand, by cutting the full state space with a hyperplane, we can divide the states into two convex blocks such that the membership problem does not require informational completeness. For instance, in the qubit case we can set $\prem = \{\varrho_{\vr} \in\state : r_z\leq 0\}$ so that the membership problem can be solved simply by measuring $\sigma_z$. On the other hand, by taking any convex subset contained in the interior of the state space we have a membership problem which, by Proposition \ref{prop:boundary}, cannot be solved without informational completeness. 
Therefore, we need to have a stronger assumption than the convexity in order to conclude something about the corresponding membership problem $\state=\prem\cupdot\prem^C$.

For a convex subset $\prem\subset \state$, the interior ${\rm Int}\, \prem$ is defined similarly as for the whole state space $\state$, that is, as the interior in the affine hull of \(\prem\).
In contrast, in the definition of the boundary the complement is taken with respect to $\prem$ rather than the full state space. In other words, we set $\partial\prem = \prem \setminus {\rm Int }\, \prem$. 
We remark that this definition of the boundary slightly differs from the usual one of e.g.~\cite[Section 6]{Rockafellar}. 
In particular, a boundary in our sense may be empty, such as in the case $\prem = {\rm Int}\, \state$.

We say that a convex subset $\prem \subset \state$ is \emph{strictly convex}  if, for all $\varrho,\sigma \in \prem$, $\varrho \neq \sigma$, and $0 < t < 1$, the mixture $t\varrho + (1-t) \sigma$ is in the interior of $\prem$. 
This is still not strong enough property to have conclusive implications for a membership problem. Namely, a strictly convex subset in the interior obviously requires informational completeness by Proposition \ref{prop:boundary}, but in general this need not be the case. As an example, suppose that $\dim \, \hi\geq 3$ and decompose the Hilbert space as a direct sum  $\hi = \hi_1\oplus\hi_2$ with $\dim\, \hi_2 =2$. For the subset $\prem$, take those states whose support is contained in $\hi_2$. Then $\prem\simeq \mathcal{S} (\hi_2)$ so that it is strictly convex, but the  membership problem $\state=\prem\cupdot\prem^C$ can be solved without informational completeness. In fact, the binary POVM $E_1=P$, $E_2=\id - P$, where $P$ is the projection onto $\hi_1$, works in this case. The essential point here is that $\prem$, or more specifically its boundary, does not contain any points from the interior of the state space. By adding this extra assumption, we obtain the following result.

\begin{proposition}\label{prop:strictly_convex_set}
Let $\prem \subset \state$ be a strictly convex set such that the boundary $\partial\prem$  contains an interior point of $\state$. 
Then the membership problem $\state=\prem \cupdot \prem^C$ cannot be solved without an informationally complete measurement.
\end{proposition}

\begin{proof}
We apply Proposition \ref{prop:learn-2}. Let $\Delta$ be a perturbation operator. 
Fix $\varrho \in \partial\prem \cap \mathrm{Int}\, \state$ and a nonzero $\lambda \in\real$ which is small enough so that $\varrho \pm \lambda \Delta \in \state$. If one of these states is in $\prem^C$ the claim follows, so we assume that $\varrho \pm\lambda \Delta \in \prem$. But then
\(
\varrho = \half (\varrho + \lambda \Delta) + \half (\varrho - \lambda \Delta) \, 
\)
which, by the strict convexity, contradicts the fact that $\varrho$ is on the boundary of $\prem$. Therefore, one of the states must be in $\prem^C$.
\end{proof}

The above result is local in nature. Indeed, the strict convexity assumption in Proposition \ref{prop:strictly_convex_set} can be relaxed as follows.

\vspace{0.5cm}

\begin{minipage}{0.4\textwidth}
\begin{center}
\includegraphics[height=0.85\linewidth]{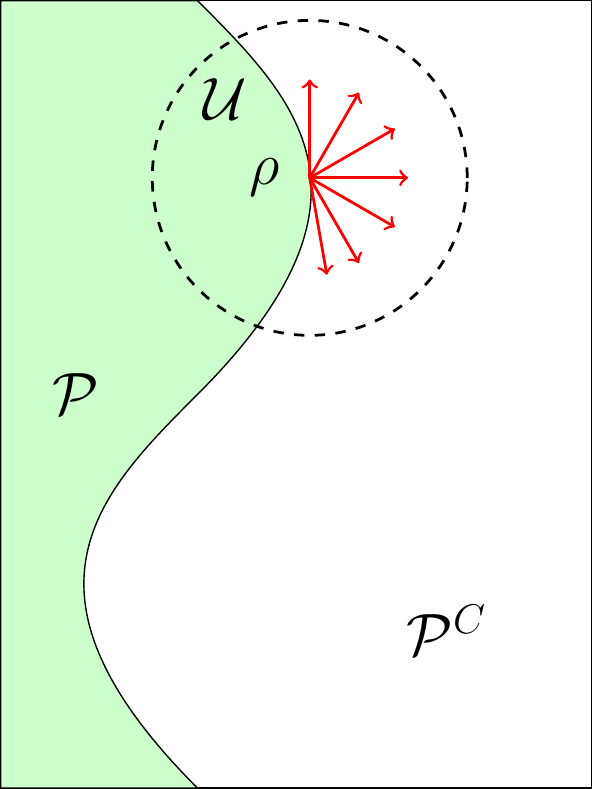}
\captionof{figure}{} \label{fig4}
\end{center}
\end{minipage}
\begin{minipage}{0.55\textwidth}
\begin{proposition}\label{prop:locally_strictly_convex_set}
Let $\prem\subset\state$, and suppose that
\begin{enumerate}
\item there exists an open set ${\mathcal U}\subset{\rm Int}\,\state$ such that the intesection $\prem\cap{\mathcal U}$ is strictly convex;
\item the boundary $\partial(\prem\cap{\mathcal U})$ is nonempty.
\end{enumerate}
Then the membership problem $\state = \prem\cupdot\prem^C$ cannot be
solved without an informationally complete measurement.
\end{proposition}
\end{minipage}

\begin{proof}
Let $\rho\in\partial(\prem\cap{\mathcal U})$. Since ${\mathcal U}$ is open, for any perturbation operator $\Delta$ there exists a suitably small nonzero $\lam\in\R$ such that $\rho\pm\lam\Delta\in{\mathcal U}$. As in the previous proof, we assume that $\rho\pm\lam\Delta\in\prem$, and then we obtain the contradiction $\rho\notin\partial(\prem\cap{\mathcal U})$ by strict convexity of $\prem\cap{\mathcal U}$.
\end{proof}

The geometrical content of proposition \ref{prop:locally_strictly_convex_set}  is illustrated by the set $\prem$ in figure \ref{fig4}.

\subsection{Sublevel set of a strictly convex function.}

We now consider another instance which captures the previously described idea of strict convexity. 
For any function $f:\state\to\R$ and any $\varepsilon\in\R$, we denote
\begin{equation*}
\state^{\leq\varepsilon}_f = \{\varrho\in\state : f(\varrho)\leq\varepsilon\} \, , \quad \state^{> \varepsilon}_f = \{\varrho\in\state : f(\varrho)>\varepsilon\} \, ,
\end{equation*}
and we consider the membership problem $\state=\state^{\leq\varepsilon}_f \cupdot \state^{> \varepsilon}_f$ for a specific class of functions.

A function $f:\state\to\R$ is called {\em strictly convex} if, for all $\varrho_1 \neq \varrho_2$ and $0 < t < 1$, we have
\begin{equation}\label{eqn:strictly_convex}
f(t \varrho_1 + (1-t) \varrho_2) < t f(\varrho_1) + (1-t) f(\varrho_2) \, .
\end{equation}
In the statement of Proposition \ref{prop:strict_convex} below the seemingly weaker hypothesis of {\em strict mid-point convexity} will suffice, which requires that the inequality \eqref{eqn:strictly_convex} holds for $t=\frac{1}{2}$. In the case of convexity (i.e., with $\leq$ instead of $<$) it is well-known that for continuous functions these two notions coincide \cite[Proposition 1.3 and the subsequent remark]{Simon2011}. Actually, it can be proved that the same is true for strict convexity.

The sublevel sets of a strictly convex continuous function defined on a strictly convex domain are themselves strictly convex. One might then be tempted to think that Proposition~\ref{prop:strictly_convex_set} implies that the membership problem $\state=\state^{\leq\varepsilon}_f \cupdot \state^{> \varepsilon}_f$ cannot be solved without an informationally complete measurement in the case of strictly convex functions. However,  the domain $\state$ {\em is not} strictly convex, so we can not conclude that the sublevel sets are strictly convex in our setting (see e.g.~Section \ref{sec:specific}.\ref{subsec:specificB} below). 
Actually, a moment's thought shows why this may not happen. For instance, consider the function \(f\colon \R^2\to \R\,,\, f(x,y)=x^2+y^2\). It is strictly convex, but, if we restrict it to the unit square $[0,1]\times[0,1]$, its sublevel sets are not. The point here is that the domain $[0,1]\times[0,1]$ is not strictly convex, hence its boundary contains nontrivial convex pieces, and such pieces can be included in the sublevel sets of $f$.

We now proceed to prove our result concerning the membership problem for the sublevel set of a strictly convex continuous function. 

\begin{proposition}\label{prop:strict_convex}
Let $f:\state\to\R$ be a strictly mid-point convex continuous function on \(\state\). If $\varepsilon\in\R$ is such that $\min f<\varepsilon<\max f$, then 
the membership problem $\state=\state^{\leq\varepsilon}_f \cupdot \state^{> \varepsilon}_f $ cannot be solved without an informationally complete measurement.
\end{proposition}

Clearly, if the function $f:\state\to\R$ is {\em strictly mid-point concave}, that is, $-f$ is strictly mid-point convex, the previous proposition can be restated for the membership problem $\state = \state^{\geq\varepsilon}_f \cupdot \state^{< \varepsilon}_f$. Note that in this case the sublevel set $\state^{\leq\varepsilon}_f$ is replaced by $\state^{< \varepsilon}_f = \{\varrho\in\state : f(\varrho)<\varepsilon\}$.

\begin{proof}
Let $\varrho_{\rm min},\varrho_{\rm max}\in\state$ be such that $f(\varrho_{\rm min}) = \min f$ and $f(\varrho_{\rm max}) = \max f$. Clearly, $\varrho_{\rm min}\in\state^{\leq\varepsilon}_f$ and $\varrho_{\rm max}\in\state^{>\varepsilon}_f$, hence the sets $\state^{\leq\varepsilon}_f$ and $\state^{>\varepsilon}_f$ are nonempty. As $\state^{\leq\varepsilon}_f\cup\state^{>\varepsilon}_f = \state$, at least one of the two sets $\state^{\leq\varepsilon}$ or $\state^{>\varepsilon}$ contains a full rank state $\varrho_0$. If $\varrho_0\in\state^{\leq\varepsilon}$, then, choosing an arbitrary $\varrho_1\in\state^{>\varepsilon}$, for all $t\in [0,1)$ the state $\varrho_t = t\varrho_1 + (1-t)\varrho_0$ has full rank. By continuity of the map $f$ and the intermediate value theorem, there exists $\bar{t}\in [0,1)$ such that $f(\varrho_{\bar{t}}) = \varepsilon$. Otherwise, if $\varrho_0\in\state^{>\varepsilon}$, then for all $t\in [0,1)$ the state $\varrho_t = t\varrho_{\rm min} + (1-t)\varrho_0$ has full rank, and continuity of $f$ and the intermediate value theorem again imply the existence of $\bar{t}\in [0,1)$ such that $f(\varrho_{\bar{t}}) = \varepsilon$. In both cases, we obtain a full rank state $\varrho_{\bar{t}}$ such that $f(\varrho_{\bar{t}}) = \varepsilon$.

For any perturbation operator $\Delta$, for sufficiently small nonzero $\lam\in\R$ we then have $\varrho_{\bar{t}}\pm \lam\Delta\in\state$. 
The remaining thing is to observe that at least one of the states $\varrho_{\bar{t}}+\lam\Delta$ or $\varrho_{\bar{t}}-\lam\Delta$ is in $\state^{>\varepsilon}_f$; indeed, since $\varrho_{\bar{t}}\in\state^{\leq\varepsilon}_f$, by Proposition \ref{prop:learn-2} this is enough to prove the thesis.
This is shown by contradiction: if $\varrho_{\bar{t}}\pm \lam\Delta\in\state^{\leq\varepsilon}_f$, then by the strict mid-point convexity we get
$$
f(\varrho_{\bar{t}}) = f\left(\frac{1}{2}(\varrho_{\bar{t}}+\lam\Delta) + \frac{1}{2}(\varrho_{\bar{t}}-\lam\Delta)\right) < \frac{1}{2} f(\varrho_{\bar{t}}+\lam\Delta) + \frac{1}{2} f(\varrho_{\bar{t}}-\lam\Delta) \leq f(\varrho_{\bar{t}}) \,.
$$
\end{proof}

As an example of the usage of Proposition \ref{prop:strict_convex}, consider the problem of deciding if a qubit state is close to a fixed reference state $\varrho_\vr $ in trace distance.
To apply Proposition \ref{prop:strict_convex}, we choose for the function $f$ the square of the trace distance, $f(\varrho) = \Vert \varrho-\varrho_\vr\Vert_1^2$. The reason for taking the square is that while the trace distance is not strictly convex, its square is. To see this, note that the trace distance between two qubit states can be expressed in the Bloch representation as $\no{ \varrho_{\va} - \varrho_{\vr} }_1 = \no{ \va - \vr }$. Since the Euclidian norm satisfies the parallellogram law, we have that for any two distinct states $\varrho_\va$ and $\varrho_\vb$,
\begin{align*}
f\left( \half\varrho_{\va}  + \half\varrho_{\vb} \right) &= \no{ \half (\varrho_\va -\varrho_\vr) + \half(\varrho_\vb - \varrho_\vr) }_1^2  = \no{\half(\va-\vr) + \half(\vb-\vr) }^2\\
&=  2\no{\half(\va-\vr)}^2 + 2\no{\half(\vb-\vr) }^2  - \no{\half(\va-\vr) - \half(\vb-\vr) }^2 \\
&= \half\no{\va-\vr}^2 + \half\no{\vb-\vr}^2  - \no{\half(\va-\vb) }^2 \\
& < \half f(\varrho_\va) + \half f(\varrho_\vb),
\end{align*}
so that $f$ is stricly mid-point convex. The $\varepsilon$-sublevel sets for the trace distance are just $\state^{\leq \varepsilon^2}_f$, so we can conclude that \emph{answering the membership problem as to whether or not a state is close to a fixed qubit state in trace norm requires informational completeness}. 

\section{Specific membership problems}\label{sec:specific}

\subsection{Norm distance membership problem.}\label{subsec:specificA}

Earlier we have seen that the trace distance membership problem for a qubit cannot be solved without an informationally complete measurement. 
The key observation was that the trace distance between states reduces to the Euclidean distance between the corresponding Bloch vectors, so that the parallellogram law holds. In a real vector space, a norm satisfies the parallellogram law if and only if it comes from an inner product. 
For this reason, we study the norm distance membership problem for the Hilbert-Schmidt norm. 

To this end, let $\sigma\in\state$ be a fixed reference state, and define $f:\state\to\R$, $f(\rho) = \Vert \rho-\sigma \Vert_2^2$. Using the parallellogram law and making the same calculation as before, we then have, for any $\varrho_1,\varrho_2\in\state$,
\begin{align*}
f\left( \half\varrho_1 + \half\varrho_2 \right) & 
 < \half f(\varrho_1) + \half f(\varrho_2)
\end{align*}
so that $f$ is strictly mid-point convex. Proposition~\ref{prop:strict_convex} then immediately tells us that whenever the two sets $\state^{\leq \varepsilon^2}_f\equiv B_\varepsilon(\sigma)$ and $\state^{>\varepsilon^2}_f\equiv B_\varepsilon(\sigma)^C$ are both nonempty, the membership problem $\state=B_\varepsilon(\sigma) \cupdot B_\varepsilon(\sigma)^C$ cannot be solved without informational completeness. Since the square of the norm is not relevant for the geometry of the sublevel sets, we can conclude that \emph{for any $0<\varepsilon <\max_{\varrho\in\state} \no{\rho-\sigma}_2$, the norm distance membership problem cannot be solved without an informationally complete measurement}.

The situation changes drastically if we consider the case $\varepsilon=0$. 
This means that our task is to say if the state is exactly the same as the reference state. If $\sigma$ is a pure state, then the answer is simple. 
By measuring just the binary POVM $E_1 = \sigma$, $E_2=\id- \sigma$ we can find out if $\varrho = \sigma$. 
Hence, the membership problem $\state=\{\sigma\}\cupdot\{\sigma\}^C$ can be solved without informational completeness.  
If, on the other hand,  $\sigma$ is in the interior of the state space,  then we need an informationally complete measurement to solve the membership problem; this follows immediately from Proposition~\ref{prop:boundary}.
In order to see that the intermediate cases $1< {\rm rank}\, \sigma<d$ do not require informational completeness, it is enough to find one perturbation operator $\Delta$ such that $\sigma + \lambda\Delta \notin \state$ for any nonzero $\lambda\in\R$. 
For this, we  write the spectral decomposition $\sigma = \sum_{j=1}^r \mu_j\vert\phi_j\rangle\langle\phi_j\vert$ where $r={\rm rank}\, \sigma$ and define 
\begin{equation*}
\Delta = \vert \phi_r \rangle\langle \phi_{r+1}\vert + \vert \phi_{r+1} \rangle\langle \phi_r\vert\in\tzero\, .
\end{equation*}
In matrix form, we then have
\begin{equation*}
\sigma + \lambda \Delta = \left[ \begin{array}{cccccc}
\mu_1 & & & & & \\
& \ddots & & & & \\
& & \mu_r & \lambda & & \\
& & \lambda & 0 & & \\
& & &  & \ddots  & \\
& & & & & 0 
\end{array}\right]
\end{equation*}
which is not positive since it has a negative minor
\begin{equation*}
\left\vert \begin{array}{cc} 
\mu_r & \lambda \\
\lambda & 0 
\end{array}\right\vert = -\lambda^2 <0
\end{equation*}
whenever $\lambda\neq 0$. 
We conclude that \emph{to decide if the state is exactly the same as a reference state $\sigma$ requires an informationally complete measurement if and only if $\sigma$ is a full rank state}.

\subsection{Fidelity membership problem.}\label{subsec:specificB}

As another example of a membership problem related to the closeness of states, we consider the fidelity with respect to a fixed reference state $\sigma$. 
The fidelity $F(\varrho,\sigma)$ of two states $\varrho$ and $\sigma$ is defined as
\begin{equation*}
F(\varrho,\sigma) = \tr{ \sqrt{ \sqrt{\varrho} \sigma \sqrt{\varrho}} } \, .
\end{equation*}
The fidelity is symmetric with respect to its arguments, and it takes its values in the interval $[0,1]$ with $1$ corresponding to the states being equal, and $0$ to the states having disjoint supports. Moreover the fidelity is a concave function, i.e.,
\begin{equation*}
F(t \varrho_1 + (1-t)\varrho_2, \sigma) \geq t F(\varrho_1, \sigma) + (1-t) F(\varrho_2, \sigma) \, .
\end{equation*}
For any  $0\leq \varepsilon \leq 1$, we  denote 
\begin{equation*}
\state^{\geq\varepsilon}_{F(\cdot,\sigma)} = \{\varrho\in\state : F(\varrho,\sigma)\geq\varepsilon\} \, , \quad \state^{ < \varepsilon}_{F(\cdot,\sigma)} = \{\varrho\in\state : F(\varrho,\sigma)<\varepsilon\} \, ,
\end{equation*}
and the membership problem is thus $\state=\state^{\geq\varepsilon}_{F(\cdot,\sigma)} \cupdot \state^{<\varepsilon}_{F(\cdot,\sigma)} $.

Let us first suppose that $\sigma$ is on the boundary of the state space. Then we can always find a perturbation operator $\Delta$ such that $\sqrt{\sigma} \Delta\sqrt{\sigma}=0$. In fact, taking the spectral decomposition \(\sigma = \sum_{j=1}^r \mu_j\vert\phi_j\rangle\langle\phi_j\vert\), and   $\Delta = \vert \phi_r \rangle\langle \phi_{r+1}\vert + \vert \phi_{r+1} \rangle\langle \phi_r\vert$  gives
\begin{equation*}
\sqrt{\sigma} \Delta\sqrt{\sigma} = \sum_{j,k=1}^r \sqrt{\mu_j\mu_k}\,  \vert \phi_j\rangle\langle\phi_j \vert \big(\vert \phi_r \rangle\langle \phi_{r+1}\vert + \vert \phi_{r+1} \rangle\langle \phi_r\vert \big) \vert\phi_k\rangle\langle\phi_k\vert =0 \, .
\end{equation*}
But with this choice we have 
\begin{equation}\label{eq:F_not_strictly_convex}
F(\varrho + \lambda\Delta, \sigma) = \tr{\sqrt{\sqrt{\sigma} (\varrho + \lambda\Delta)\sqrt{\sigma}}} = \tr{\sqrt{\sqrt{\sigma} \varrho \sqrt{\sigma}}}  = F(\varrho, \sigma)
\end{equation}
for all $\varrho\in\state$. In particular, $\varrho + \lambda\Delta \in \state^{\geq\varepsilon}_{F(\cdot,\sigma)} $ for all $\varrho\in\state^{\geq\varepsilon}_{F(\cdot,\sigma)} $ so that by Proposition~\ref{prop:learn-2} \emph{the membership problem $\state^{\geq\varepsilon}_{F(\cdot,\sigma)} \cupdot \state^{<\varepsilon}_{F(\cdot,\sigma)} $ can be solved without an informationally complete measurement.}

As an interesting consequence of \eqref{eq:F_not_strictly_convex} with $\rho\in{\rm Int}\,\state$, we have that, although the fidelity is a concave function, it is not strictly concave when $\sigma$ is a boundary state.

The previous argument does not work when $\sigma$ is in the interior of the state space, and this case requires a separate inspection.
From spectral calculus and \cite[Theorem 2.10]{Carlen2010} we know that the trace functional \(A\mapsto \tr{\sqrt{A}}\) is continuous and strictly mid-point concave on the set of the positive operators $\mathcal{L}_+(\hh)\subset\lhs$. If \(\sigma\) is an interior state, then the affine map \(\varrho\mapsto \sqrt{\sigma}\varrho\sqrt{\sigma}\) is injective on $\state$, hence the composition \(\varrho \mapsto- F(\varrho,\sigma)\) satisfies the hypotheses of Proposition~\ref{prop:strict_convex}. 
Therefore, the membership problem $\state^{\geq\varepsilon}_{F(\cdot,\sigma)} \cupdot \state^{<\varepsilon}_{F(\cdot,\sigma)}$ cannot be solved without an informationally complete measurement set-up.


\subsection{Purity membership problem.}\label{subsec:specificC}

Suppose that our task is to decide if an unknown state  $\varrho$ is pure or mixed.
The corresponding membership problem is hence $\state=\state_{\rm pure} \cupdot \state_{\rm mixed}$. Our first observation is that a measurement is capable of solving this membership problem if and only if it can uniquely determine every pure state among all states. Indeed, the sufficiency of the latter property is obvious. To prove its necessity, let $E$ be a POVM such that there exist a pure state $\varrho_1$ and another state $\varrho_2\neq\varrho_1$ such that $\tr{\varrho_1E_j}=\tr{\varrho_2E_j}$ for all $j$. 
Then
\begin{equation*}
\tr{\left(\half\varrho_1 + \half\varrho_2\right) E_j} = \half\tr{\varrho_1 E_j} + \half\tr{\varrho_2E_j} = \tr{\varrho_1E_j}
\end{equation*}
for all $j$, which shows that the mixed state $\half\varrho_1 + \half\varrho_2$ cannot be distinguished from the pure state $\varrho_1$. In other words, $E$ fails to solve the membership problem. This means that if a measurement can solve the membership problem, then the measurement outcome distribution corresponding to a pure state determines the state uniquely among all states.

It is known that for dimensions $d=2$ and $d=3$ a collection of measurements that gives a unique measurement statistics to every pure state among all states must be informationally complete \cite{CaHeScTo14}.
Therefore, for qubits and qutrits, the membership problem cannot be solved without an informationally complete measurement. 
We can see this directly by observing that an arbitrary perturbation operator $\Delta$ can be written, up to a scaling, as a difference of a pure state and a mixed state.
In the qubit case, if $U$ is a unitary operator that diagonalizes $\Delta$, we have 
\begin{align*}
U\Delta U^* = \left[\begin{array}{cc}\lambda & 0 \\ 0 & -\lambda \end{array}\right] =-2 \lambda  \left( \left[\begin{array}{cc}0 & 0 \\ 0 & 1 \end{array}\right]- \left[\begin{array}{cc}\half & 0 \\ 0 & \half \end{array}\right] \right) \, ,
\end{align*}
which is a difference of a mixed state and pure state. In the qutrit case this fact follows from the previous equation whenever $\rank\Delta = 2$, and from
\begin{align*}
U\Delta U^* = \left[\begin{array}{ccc}\lambda & 0 & 0 \\0 & \mu & 0 \\0 & 0 & -\lambda -\mu\end{array}\right] = -(\lambda + \mu)\left( \left[\begin{array}{ccc}0 & 0 & 0 \\0 & 0 & 0 \\0 & 0 & 1\end{array}\right] -  \left[\begin{array}{ccc}\frac{\lambda}{\lambda+\mu} & 0 & 0 \\0 & \frac{\mu}{\lambda + \mu} & 0 \\0 & 0 & 0\end{array}\right] \right)
\end{align*}
with $\lam,\mu>0$ when $\rank\Delta = 3$.

In dimensions $d\geq 4$, there are measurements that solve the membership problem $\state=\state_{\rm pure} \cupdot \state_{\rm mixed}$, or, equivalently, give a unique measurement statistics to every pure state among all states, but fail to be informationally complete \cite{ChDaJietal13}. 
To see this, we choose a perturbation operator
\begin{equation*}
\Delta = \kb{\phi_1}{\phi_1} + \kb{\phi_2}{\phi_2} -\kb{\phi_3}{\phi_3} -\kb{\phi_4}{\phi_4} \, , 
\end{equation*}
where $\{\phi_j\}_{j=1}^d$ is an orthonormal basis.
By Weyl's inequalities \cite[Theorem III.2.1]{Bhatia}, $\lam\Delta$ cannot be written as a difference of a pure state and a mixed state for any $\lam\in\R$. Therefore, any measurement with $\rr^\perp=\R\Delta$ solves the membership problem $\state=\state_{\rm pure} \cupdot \state_{\rm mixed}$, but is not informationally complete as it can not distinguish the mixed states $(\kb{\phi_1}{\phi_1} + \kb{\phi_2}{\phi_2})/2$ and $(\kb{\phi_3}{\phi_3} + \kb{\phi_4}{\phi_4})/2$.

We can try to relax the problem by asking if the unknown state is almost pure. There are two natural ways to quantify this: the purity $f_1(\varrho)=\tr{\varrho^2}$ and the von Neumann entropy $f_2(\varrho)=-\tr{\varrho \log_2\varrho}$. Purity takes its values in the interval $[1/d,1]$ with $1$ corresponding to a pure state, whereas von Neumann entropy is in $[0,\log_2 d]$ with pure states giving the value $0$. 
However, the purity is strictly convex since it is just the square of the Hilbert-Schmidt norm. 
Also, the von Neumann entropy is strictly concave \cite[Theorem 2.10]{Carlen2010}, hence the function $-f_2$ is strictly convex. 
As a direct consequence of Proposition~\ref{prop:strict_convex}, we conclude that the \emph{membership problems $\state=\state^{\leq\varepsilon}_{f_1}\cupdot \state^{>\varepsilon}_{f_1}$ and $\state=\state^{\geq\varepsilon}_{f_2}\cupdot \state^{<\varepsilon}_{f_2}$ cannot be solved without an informationally complete measurement.}

\subsection{Rank membership problem.}
\label{subs:finalf}
Since in dimensions $d\geq 4$ it is possible to decide if a state is pure without an informationally complete measurement, it makes sense to ask if we can also determine the rank of the state without identifying it. 
We denote
$
\state_{\rm rank}^{r} = \{\varrho\in\state : {\rm rank}\, \varrho =r\}
$
and the membership problem is hence $\state=\bigcupdot_{r=1}^d \state_{\rm rank}^r$. However, since $\state_{\rm rank}^d$ is the interior of $\state$, Proposition~\ref{prop:boundary} immediately tells us that the membership problem cannot be solved without informational completeness. 
We can relax the problem and consider whether or not the rank of the state is below some value $r\leq d-1$, i.e., the membership problem is then $\state=\state^{\leq r}_{\rm rank}\cupdot \state^{>r}_{\rm rank}$. 
As we will see, the solution of this membership problem depends on the bound $r$.

To this aim, we first recall that for a perturbation operator $\Delta$, we write $\Delta = \Delta_+ - \Delta_-$, where $\pm\Delta_\pm$ are the positive and negative parts of $\Delta$ as given by the spectral decomposition.
It follows that $\rank\Delta_+ + \rank\Delta_- =\rank\Delta$, and this decomposition is minimal in ranks in the following sense:
whenever $\Delta=\lambda (\varrho_1-\varrho_2)$ for some $\lambda>0$ and $\varrho_1,\varrho_2\in\state$, then $\rank\varrho_1\geq \rank\Delta_+$ and $\rank\varrho_2\geq \rank\Delta_-$; see e.g. \cite[Lemma 1.(c)]{CaHeScTo14} for a proof.

We now split the treatment into two cases.
Let us first consider the case when the bound $r<\floor{d/2}$, where $\floor{d/2}$ denotes the largest integer not greater than $d/2$.  
We choose a perturbation operator $\Delta$ such that ${\rm rank}\, \Delta_\pm = \floor{d/2}>r$. 
By the earlier remark, then in any decomposition $\Delta=\lambda (\varrho_1-\varrho_2)$ we always have $\varrho_1, \varrho_2 \in \state_{\rankbis}^{>r}$. 
Hence, Theorem \ref{thm:learn} implies that in this case the membership problem $\state=\state^{\leq r}_{\rm rank}\cupdot \state^{>r}_{\rm rank}$ can be solved without an informationally complete measurement.

We now move to the case  $r\geq \floor{d/2}$, and show that in this case informational completeness is needed. Since for any perturbation operator $\Delta=\Delta_+ - \Delta_-$ we have $\rank\Delta_+ + \rank\Delta_-\leq d$, we must have $\rank\Delta_+\leq r$ or $\rank\Delta_-\leq r$. If one of the ranks, say $\rank\Delta_-$, is strictly greater than $r$, then we can define $\varrho = \Delta_- /\tr{\Delta_-}\in\state_{\rankbis}^{>r} $ so that $\varrho + \Delta/\tr{\Delta_-} = \Delta_+/\tr{\Delta_-} $ is a state and
$$
\rank(\varrho + \Delta/\tr{\Delta_-})  = \rank(\Delta_+/\tr{\Delta_-})\leq r \, .
$$
Hence, $\varrho + \Delta/\tr{\Delta_-}\in\state_{\rankbis}^{\leq r}$.  If, on the other hand, both $\rank\Delta_\pm\leq r$ then we have two possibilities. If $\rank{\abs{\Delta}} >r$ then we set $\varrho = \abs{\Delta}/\tr{\abs{\Delta}}\in\state_{\rankbis}^{ >r}$ and we have
$$
\varrho + \Delta/\tr{\abs{\Delta}} = 2\Delta_+ /\tr{\abs{\Delta}}  \in \state_{\rankbis}^{\leq r}\, .
$$
If $\rank\abs{\Delta} \leq r$ then, since $r\leq d-1$, we can pick an orthogonal projection \(P\) onto a \((r+1-\rank\abs{\Delta})\)-dimensional subspace orthogonal to \({\rm supp}\,\abs{\Delta}\). For such a projection, we have 
\begin{equation*}
\rank(\abs{\Delta} + P) = r+1>r
\end{equation*}
but 
\begin{equation*}
\rank(2\Delta_+ + P) = \rank(\Delta_+ + P) = \rank(\abs{\Delta} + P) - \rank\Delta_- \leq r \, .
\end{equation*}
We can then set
$$
\varrho = (\abs{\Delta} + P)/\tr{\abs{\Delta} + P} \in \state_{\rankbis}^{>r} \, ,
$$
so that we have 
$$
\varrho + \Delta/\tr{\abs{\Delta} + P} = (2\Delta_+ + P)/\tr{\abs{\Delta} + P} \in\state_{\rankbis}^{\leq r} \,.
$$
In all the cases, we have found a state $\rho\in\state_{\rankbis}^{>r}$ and $\lam\in\R$ such that $\rho + \lam\Delta\in\state_{\rankbis}^{\leq r}$. By Proposition~\ref{prop:learn-2} we conclude that the membership problem $\state=\state^{\leq r}_{\rm rank}\cupdot \state^{>r}_{\rm rank}$ with $r \geq\floor{d/2}$ cannot be solved without an informationally complete measurement set-up, as claimed. 

Summing up, \emph{the membership problem $\state=\state^{\leq r}_{\rm rank}\cupdot \state^{>r}_{\rm rank}$ can be solved without an informationally complete measurement if and only if $r <\floor{d/2}$}.

\section{Minimal number of outcomes}\label{sec:minimal}

\subsection{General formulation}\label{subsec:general-min}

Whenever a membership problem can be solved without an informationally complete measurement, it is natural to ask what are the optimal measurements for solving the problem. One possible way to quantify optimality is in terms of the minimal number of measurement outcomes: we look for a POVM consisting of as few operators as possible while retaining the ability to solve the problem. This question has recently been studied extensively for various tomographic problems where prior information is exploited to reduce the number of outcomes \cite{HeMaWo13,ChDaJietal13,CaHeScTo15,KeWo15,KeVrWo15,CaHeKeScTo16,BaDeKa16,Ma16,Kech16}. 
In this section we make some observations on the membership problems studied in Section \ref{sec:specific}.

Recall that from the point of view of state distinguishability, it is the operator system $\rr$ generated by the POVM that is the relevant mathematical object. Although the same operator system can be generated by many different POVMs, there always exist generating POVMs having $\dim \rr$ different elements by \cite[Proposition 1]{HeMaWo13}. Clearly, no POVM with less elements can generate the same operator system. Hence, the minimal number of measurement outcomes for a given membership problem is given by the minimal dimension of the operator system which succeeds in the task. Furthermore, since $\dim\lhs =d^2 = \dim \rr + \dim \rr^\perp$, we can equivalently search for the maximal dimension of the orthogonal complement $\rr^\perp$. 

In Section \ref{sec:membership} we noted that an operator system $\rr$ solves the membership problem $\state=\cupdot_j \prem_j$ if and only if for any $\varrho_j\in\prem_j$ and $\varrho_k\in\prem_k$ with $j\neq k$, we have $\varrho_j-\varrho_k\notin \rr^\perp$. 
This can be restated as
\begin{equation}\label{eq:empty}
\left[ \bigcup_{j\neq k } (\prem_j - \prem_k) \right] \cap \rr^\perp = \emptyset \, , 
\end{equation}
and the question hence boils down to finding the maximal dimension of a subspace $\rr^\perp\subset\tzero$ which satisfies \eqref{eq:empty}.

\subsection{State identification.} 

Let $\sigma\in\state$ be a fixed reference state with $r = \rank\sigma <d$, and consider the problem of determining if an unknown state $\varrho\in\state$ is equal to $\sigma$, i.e., the membership problem is $\state=\{\sigma\}\cupdot\{\sigma\}^C$. 
As we have seen earlier, this can be solved without informational completeness (see Section \ref{sec:specific}.\ref{subsec:specificA}).
Further, it is relatively easy to calculate the minimal number of outcomes. To this end, note that $\sigma$ is on the face $\mathcal{F}(\sigma) = \{ \varrho\in\state : {\rm supp}\, \varrho \subseteq {\rm supp}\, \sigma \}$ of the state space and we have $\mathcal{F}(\sigma)  \simeq \state(\C^r)$. We can thus view $\sigma$ as an interior point of $\state(\C^r)$, and therefore distinguishing a state $\varrho\in\mathcal{F}(\sigma)$ from $\sigma$ requires ``informational completeness on the face $\mathcal{F}(\sigma)$''. In other words, the minimal number of outcomes is at least $r^2$.

On the other hand, suppose that we have $r^2$ linearly independent positive operators $E_1,\ldots, E_{r^2}$ which sum up to the projection onto ${\rm supp}\, \sigma$, that we denote by $Q$. By adding one more element $E_{r^2+1} = \id-Q$, which corresponds to verifying if the state $\varrho$ is actually on the face $\mathcal{F}(\sigma)$, we have constructed a POVM with $r^2+1$ outcomes which solves the membership problem $\state=\{\sigma\}\cupdot\{\sigma\}^C$.

In order to show that $r^2+1$ (and not $r^2$) is actually the minimal number of outcomes, it is enough to prove that the set $\R(\{\sigma\}-\{\sigma\}^C)$ contains a $r^2$-dimensional linear subspace $\xx\subset\tzero$.  Indeed, \eqref{eq:empty} then implies that $\xx\cap\rr^\perp = \{0\}$ for any $\rr$ that solves the membership problem. 
Hence $\dim\xx + \dim\rr^\perp \leq \dim\tzero = d^2-1$, implying $\dim\rr\geq r^2+1$. 
So, we fix a state $\tau$ with ${\rm supp}\,\tau \nsubseteq {\rm supp}\,\sigma$ and denote
$$
\xx = \{\lam[\sigma - (t\varrho + (1-t)\tau)] : \lam\in\R,\ t\in[0,1],\ \varrho\in\mathcal{F}(\sigma)\} \, , 
$$
and we claim that the set $\xx$ is as stated above.
For the linear combination of elements of $\xx$, we have
\begin{equation}\label{eq:linear_comb}
\lam_1[\sigma - (t_1\varrho_1 + (1-t_1)\tau)] + \lam_2[\sigma - (t_2\varrho_2 + (1-t_2)\tau)] = \Delta + \mu(\sigma-\tau)
\end{equation}
where $\Delta = \lam_1 t_1 (\sigma-\varrho_1) + \lam_2 t_2 (\sigma-\varrho_2)$ and $\mu = \lam_1 (1-t_1) + \lam_2 (1-t_2)$. If $\Delta = 0$, \eqref{eq:linear_comb} is clearly an element in $\xx$. Otherwise, $\Delta$ is a perturbation operator with ${\rm supp}\,\Delta\subseteq {\rm supp}\,\sigma$, hence it is given by $\lam(\sigma-\varrho_\lam)$ for sufficiently large $\abs{\lam}\in\R$ and $\varrho_\lam=\sigma-\lam^{-1}\Delta\in\mathcal{F}(\sigma)$. The linear combination \eqref{eq:linear_comb} is then
$$
\lam(\sigma-\varrho_\lam) + \mu(\sigma-\tau) \,,
$$
and to prove that it is still an element of $\xx$, it suffices to show that
$$
\lam = \lam't \quad \text{and} \quad \mu = \lam'(1-t) \quad \text{for some} \quad \lam'\in\R,\ t\in [0,1]
$$
for sufficiently large $\abs{\lam}$. Solving the latter equations, we get $\lam' = \lam+\mu$ and $t=\lam/(\lam+\mu)\in [0,1]$ if $\lam$ is chosen with the same sign as $\mu$. This completes the proof that $\xx$ is a linear space. Since $\xx$ is generated by the linear subspace
$$
\tzero(\sigma) = \{\lam(\sigma-\varrho) : \lam\in\R,\ \varrho\in\mathcal{F}(\sigma)\} = \{\Delta\in\tzero : {\rm supp}\,\Delta\subseteq {\rm supp}\,\sigma\}
$$
together with the operator $\sigma-\tau\in\tzero\setminus\tzero(\sigma)$, it follows that $\dim\xx = \dim\tzero(\sigma) + 1 = r^2$, as claimed.

\subsection{Fidelity separation.}

Let $\sigma\in\state$ again be a fixed reference state on the boundary of the state space, and consider the fidelity membership problem $\state=\state^{\geq\varepsilon}_{F(\cdot,\sigma)} \cupdot \state^{<\varepsilon}_{F(\cdot,\sigma)} $. 
As discussed earlier, the perturbation operators $\Delta$ with $\sqrt{\sigma}\Delta \sqrt{\sigma} =0$  satisfy $F(\varrho + \Delta, \sigma) = F(\varrho,\sigma)$. Therefore, if we set $\rr^\perp = \{\Delta \in\tzero : \sqrt{\sigma}\Delta \sqrt{\sigma} =0 \}$ we have  $(\state^{\geq\varepsilon}_{F(\cdot,\sigma)} - \state^{<\varepsilon}_{F(\cdot,\sigma)}) \cap \rr^\perp=\emptyset$ and  thus $\rr$ is an operator system which solves the membership problem. 
By computing the dimension of $\rr^\perp$ we then have an upper bound on the minimal number of outcomes for the fidelity membership problem. 

By using the spectral decomposition $\sigma = \sum_{j=1}^r \mu_j \vert\phi_j \rangle \langle \phi_j\vert$ where $r = \rank\sigma$, we have that
\begin{equation*}
\sqrt{\sigma}\Delta \sqrt{\sigma}  = \sum_{j,k=1}^r \sqrt{\mu_j\mu_k} \langle \phi_j \vert \Delta \phi_k\rangle \vert \phi_j \rangle \langle \phi_k \vert =0
\end{equation*}
if and only if $\langle \phi_j \vert \Delta \phi_k \rangle = 0$ for all $j,k=1,\ldots, r$. Since $\Delta^*=\Delta$ and $\tr{\Delta}=0$, a simple calculation shows that there are still $d^2-r^2-1$ free real parameters in $\Delta$. In other words, $\dim\, \rr^\perp= d^2 - r^2 - 1$. 
It follows that
\begin{equation*}
\dim\, \rr = d^2 - \dim\, \rr^\perp = r^2+1
\end{equation*}
which gives us an upper bound for the minimal number of outcomes for the membership problem $\state=\state^{\geq\varepsilon}_{F(\cdot,\sigma)} \cupdot \state^{<\varepsilon}_{F(\cdot,\sigma)} $.

\subsection{Rank separation.}

We have already noted in Section \ref{sec:specific}.\ref{subsec:specificC} that any measurement set-up which is capable of telling if a state is pure or not must distinguish any pure state from all other states. 
We will now prove that the rank membership problem $\state=\state^{\leq r}_{\rm rank}\cupdot \state^{>r}_{\rm rank}$, with $r\leq d-1$, shows a similar effect, namely, that a measurement set-up which solves this problem must distinguish any state in $\state^{\leq r}_{\rm rank}$ from any other state.

To see this, let $\rr$ be an operator system which solves the membership problem. It is enough to show that $\rr$ distinguishes every state $\varrho_1\in\state^{\leq r}_{\rm rank}$ from all other states $\varrho_2\in\state^{\leq r}_{\rm rank}$, as for $\varrho_2\in\state^{> r}_{\rm rank}$ this is already implied by the starting assumption.
So, suppose by contradiction that \(\rho_1-\rho_2\in\rr^\perp\) for some \(\rho_1,\rho_2\in\state^{\leq r}\); we are then going to prove that there exist \(\rho\in  \state^{\leq r}\) and \(\sigma\in \state^{> r}\) such that \(\rho_1-\rho_2 = \lambda(\rho-\sigma)\). This shows that \(\rr\) does not distinguish between $\rho$ and $\sigma$, and thus gives the desired contradiction. For $\Delta = \rho_1-\rho_2$, notice that
\[
\Delta=\abs{\Delta}-2\Delta_-
\]
where \(-\Delta_-\) is the negative part of \(\Delta\). We have the inequalities \(\rank\Delta_- <\rank\abs{\Delta}\) and \(\rank\Delta_-\leq r\). The latter inequality follows from the relation \(\Delta=\rho_1-\rho_2\) and \cite[Lemma 1.(c)]{CaHeScTo14}.  If \(\rank\abs{\Delta}> r\) we are done, since we have
\[
\Delta=\tr{\abs{\Delta}}\left( \rho-\sigma \right)
\]
with \(\rho=\abs{\Delta}/\tr{\abs{\Delta}}\) and \(\sigma=2\Delta_- /\tr{\abs{\Delta}}\). Otherwise, if $\rank\abs{\Delta}\leq r$, then we can pick an orthogonal projection \(P\) onto a \((r+1-\rank\abs{\Delta})\)-dimensional subspace orthogonal to \({\rm supp}\,\abs{\Delta}\). We have
\[
\rank(2\Delta_- + P)< \rank(\abs{\Delta}+P)=r+1 \,,
\]
and we can write 
\(
\Delta=-\tr{\abs{\Delta}+P}\left( \rho-\sigma \right)
\)
with
\begin{align*}
\varrho = (2\Delta_- +P) /\tr{\abs{\Delta}+P} \in\state^{\leq r} \, , \quad \sigma = (\abs{\Delta}+P)/\tr{\abs{\Delta}+P} \in\state^{>r} \,.
\end{align*}

We conclude that the search for the minimal number of outcomes for the membership problem $\state=\state^{\leq r}_{\rm rank}\cupdot \state^{>r}_{\rm rank}$ is equivalent to the problem for measurements distinguishing states in $\state^{\leq r}_{\rm rank}$ from all other states. 
In \cite[Theorem 2]{ChDaJietal13} it was shown that there exist $4r(d-r) + d-2r -1$ selfadjoint operators which succeed in the task, and this corresponds to $4r(d-r) + d-2r$ POVM outcomes by \cite[Proposition 3]{HeMaWo13}. This is therefore an upper bound for the minimal number of outcomes for a POVM that solves the rank separation membership problem. Accordingly to what we have shown in Section~\ref{subs:finalf}, when \(r\geq \floor{\frac{d}{2}}\) the above upper bound becomes trivial.

\section{Conclusion}\label{sec:conclusion}
Any partitioning of the quantum state space defines the corresponding membership problem, namely, the problem of deciding to which subset a given unknown state belongs to. Based on the amount of information (number of measurement outcomes) needed to solve them, the difficulty of the membership problems can be compared. In this sense, the membership problems that require an informationally complete measurement form the  class of problems  most difficult to solve. Here we have proved various general criteria of a geometric nature, which force a membership problem to fall into this class. We have also dealt with different explicit examples such as the problem of deciding if a prepared state is close to a target state, or if the rank of an unknown state is below a given bound. For the cases that do not require informational completeness, we have studied in more detail the minimal number of measurement outcomes needed for solving the membership problem. 

Our approach is based on performing individual measurements on single copies of a quantum state. As such, it should be compared with the collective measurement approach, where instead \(N\) identical copies of the same state are subject to a global measurement on the tensor product space. A survey on quantum property detection by means of collective measurement is provided e.g.~in \cite[Section 4]{MW2016} (see also the references therein); applications to entanglement detection can be found in \cite{ADH08} and \cite{ZCO11}. We remark that all membership problems involving sublevel sets of polynomial functions of the state can be solved within the collective framework \cite{Br04}.

\section*{Acknowledgments}
JS acknowledges financial support from the EU through the Collaborative Projects QuProCS (Grant Agreement No. 641277).
TH acknowledges financial support from the Academy of Finland (Project No. 287750).\\
{TH and JS thank Michael Kech for inspiring discussions on the topic of state determination during his visit in Turku.}


\end{document}